\documentclass[runningheads, envcountsame, a4paper]{llncs}

\usepackage{epstopdf}
\usepackage[hyphens]{url}
\usepackage{microtype}
\usepackage[usenames]{color}
\usepackage{amssymb}
\usepackage{graphicx}
\usepackage{amscd}

\usepackage[colorlinks=true,
linkcolor=webgreen,
filecolor=webbrown,
citecolor=webgreen]{hyperref}

\definecolor{webgreen}{rgb}{0,.5,0}
\definecolor{webbrown}{rgb}{.6,0,0}

\usepackage{color}
\usepackage{float}

\usepackage{amsmath}

\usepackage{url}
\usepackage{breakurl}

\def\derives{\Longrightarrow}
\def\derivestar{\Longrightarrow^*}
\def\union{\, \cup \,}
\def\intersect{\, \cap \,}
\DeclareMathOperator{\lexleast}{lexlt}
\DeclareMathOperator{\cyc}{cyc}
\DeclareMathOperator{\quo}{quo}
\def\Que{{\mathbb{Q}}}

\def\Enn{{\mathbb{N}}}
\def\atopp#1#2{{#1 \atopwithdelims[] #2}}

\title{Undecidability and Finite Automata}
\toctitle{Undecidability and Finite Automata}
\authorrunning{J. Endrullis, J. Shallit, and Tim Smith}
\tocauthor{J\"org~Endrullis, Jeffrey~Shallit, and Tim~Smith}

\author{J\"org Endrullis\inst{1} \and 
Jeffrey Shallit\inst{2} \and Tim Smith\inst{2}} 

\institute{
Vrije Universiteit Amsterdam,
Department of Computer Science,
De Boelelaan 1081a,
1081 HV Amsterdam,
The Netherlands \\
\email{joerg@endrullis.de}
\and
School of Computer Science,
University of Waterloo,
Waterloo, ON  N2L 3G1,
Canada\\
\email{shallit@uwaterloo.ca},
\email{timsmith@uwaterloo.ca}
}

\begin{document}

\maketitle

\begin{abstract}
Using a novel rewriting problem,
we show that several natural decision problems about finite automata
are undecidable (i.e., recursively unsolvable).
In contrast, we also prove three related problems are decidable.
We apply one result to
prove the undecidability of
a related problem about $k$-automatic sets of rational numbers.
\end{abstract}

\section{Introduction}
\label{one}

Starting with the first 
result of Turing \cite{Turing:1936}, computer scientists 
have assembled a large collection of natural
decision problems that
are undecidable (i.e., recursively unsolvable); see, for
example, the book of
Rozenberg and Salomaa \cite{Rozenberg&Salomaa:1994}.

Although some of these results deal with relatively weak computing
models, such as pushdown automata 
\cite{Bar-Hillel&Perles&Shamir:1961,Ginsburg&Rose:1963},
few, if any, are concerned
with the very simplest model:  the finite automaton.  One exception
is the following decision problem, due to Engelfriet and
Rozenberg \cite[Theorem 15]{Engelfriet&Rozenberg:1980}:
given a finite automaton $M$ with an input alphabet of
both primed and unprimed letters (i.e., an alphabet
$\Sigma \ \cup \ \Sigma'$, where $\Sigma' = \{ a' \ : \ a \in \Sigma \}$),
decide if $M$ accepts a word $w$ where the primed letters, after the
primes have been removed, form a word
identical to that formed by the unprimed letters.  This
problem was also mentioned by Hoogeboom \cite{Jan:2012}.
This problem is
easily seen to be undecidable, as it is a disguised version of the
classical Post correspondence problem \cite{Post:1965}.

In this paper we start by proving a novel lemma on rewriting systems.
In Sect.~\ref{three},
this lemma is then applied to give a new example of a natural
problem on finite automata that is undecidable.  In Sect.~\ref{four}
we prove that a related problem on the so-called $k$-automatic
sets of rational numbers is also undecidable.  In Sect.~\ref{five}
we prove the undecidability of yet another problem about finite automata.
Finally,
in Sect.~\ref{five} we show that it is decidable if a finite automaton
accepts two distinct conjugates.

\section{A lemma on rewriting systems}
\label{two}

For our purposes, a {\it rewriting system} $S$ over an alphabet
$\Sigma$ consists of 
a finite set of context-free
rules of the form $\ell \rightarrow r$, where $\ell, r \in
\Sigma^*$.  Such a rewriting rule applies to the word
$\alpha \ell \beta \in \Sigma^*$, and converts it to
$\alpha r \beta$.  We indicate this by writing
$\alpha \ell \beta \derives \alpha r \beta$.
We use
$\derivestar$ for the reflexive, transitive closure
of $\derives$ (so that $\gamma \derivestar \zeta$ means that
there is a sequence of $0$ or more rules taking the word
$\gamma$ to $\zeta$).
A rewriting system is said to be
{\it length-preserving} if $|\ell| = |r|$ for all rewriting
rules $\ell \rightarrow r$.

Many undecidable decision problems related to rewriting systems are known
\cite{Book&Otto:1993}.
However, to the best of our knowledge, the following one is
new.

\bigskip

\noindent\texttt{REWRITE-POWER}

\medskip

\noindent{\it Instance:} an alphabet $\Sigma$ containing the symbols $a$ and $b$ (and possibly other symbols), and a length-preserving rewriting system $S$.

\medskip

\noindent{\it Question:} does there exist an integer
$n \geq 1$ such that $a^n \Longrightarrow^* b^n$?

\begin{lemma} 
The decision problem \texttt{\textup{REWRITE-POWER}} is undecidable.
\end{lemma}

\begin{proof}
The standard approach for showing that a rewriting problem is
undecidable is to reduce from the halting problem, by encoding a Turing
machine $M$ and simulating its computation using the rewrite rules;
for example, see \cite{Book&Otto:1993}.

The difficulty with applying that approach in the present case is the
lack of asymmetry, i.e., the fact that the initial word consists of
all $a$'s.  Because there is no distinguished symbol with which to start
the simulation, unwanted parallel simulations of $M$ could occur at
different parts of the word.

To deal with this difficulty, we construct a rewriting system that
permits multiple simulations of $M$ to arise, but employs a delimiter 
symbol \$ to ensure that they do not interfere with each other.  Each
simulation works on its own portion of the word, and changes it to
$b$'s (ending by changing the delimiter symbol as well) only if $M$
halts.

Here are the details.  We use
the one-tape model of Turing machine from 
Hopcroft and Ullman \cite{Hopcroft&Ullman:1979},
where $M = (Q, \Omega, \Gamma, \delta, q_0, {\tt B}, q_f)$.
Here $Q$ is the set of states of $M$, with $q_0 \in Q$ 
the start state and $q_f \in Q$ the unique final state.
Let $\Omega$ be the input alphabet and
$\Gamma$ the tape alphabet,
with ${\tt B} \in \Gamma$ being the distinguished blank symbol.
Let $\delta$ be the (partial) 
transition function, with domain $Q \times \Gamma$
and range $Q \times \Gamma \times \{L, R\}$.
We assume without loss of generality that $M$ halts, i.e.,
$M$ has no next move, iff it is in state $q_f$.  We also assume that
$a, b, \$ \not\in \Gamma$.

We construct our length-preserving rewriting system mimicking the computations
of $M$ as follows.
Let $\Sigma = \Gamma \cup Q \cup \{a,b,\$\}$.  Let $S$ contain the following rewrite rules: 

\newpage

\begin{align}
	aa &\rightarrow \$ q_0 \\
	a &\rightarrow {\tt B} \\
	q_i c &\rightarrow d q_j \quad \text{ if } (q_j, d, R) \in \delta(q_i, c) \\
	f q_i c &\rightarrow q_j f d	\quad\text{ if } (q_j, d, L)] \in \delta(q_i, c) \text{ and } f \in \Gamma \\
	q_f c &\rightarrow c q_f \quad \text{ for all } c \in \Gamma \\
	c q_f &\rightarrow q_f b \quad \text{ for all } c \in \Gamma \\
	\$ q_f &\rightarrow bb
\end{align}
Rule (1) starts a new simulation.  Each simulation has its own head
symbol $q_i$ and left end-marker \$, and never moves the head
symbol past an $a$, $b$, or \$.  This ensures that the simulations are
kept separate from each other.  Rule (2) converts an $a$ to a blank
symbol, available for use in a simulation.  Rules (3) and (4) are used
to simulate the transitions of $M$.  Once a simulation reaches $q_f$
(meaning that $M$ has halted), its head can be moved to the right using
rule (5), past all of the symbols it has read, and then back to the
left using rule (6), changing all of those symbols to $b$'s.  Finally,
the simulation can be stopped using rule (7).

We argue that $M$ halts when run on a blank tape iff $a^n
\Longrightarrow^* b^n$ for some $n \geq 1$.

If $M$ halts when run on a blank tape, then there is some number of
tape squares $k$ which it uses.  Let $n = k+2$.  Then $S$ can use rule
(1) to start a simulation with the two $a$'s at the left end of the
word, use rule (2) to convert the $k$ remaining $a$'s to blank
squares, and run the simulation using rules (3) and (4).  Eventually
$M$ halts in the final state $q_f$.  At that point, $S$ can move
the head to the right end of the word using rule (5), convert all $k$
tape squares to $b$'s using rule (6), and then convert the end-marker
\$ and tape head to $b$'s with rule (7).  Thus we have $a^n
\Longrightarrow^* b^n$.

For the other direction, if the initial word $a^n$ is ever
transformed into $b^n$, it means that one or more simulations were run,
each of which operated on a portion of the word without interference
from the others.  
The absence of interference can be deduced from the shape of the rules.
There is only one occurrence of the marker \$ in the left-hand sides of
the rules, namely as the leftmost symbol in the left-hand side of rule
(7).  Furthermore, \$ can only be rewritten to $b$ which does not occur
in any left-hand side.

Each simulation runs $M$ on a blank tape, and uses a
number of tape squares bounded by the length of its portion of the
word.  Since every portion of the word was transformed into $b$'s,
$M$ halted in every one of the simulations (otherwise the word would
still contain one or more head symbols and end-markers).  The
completion of any one of these simulations is enough to show that $M$
halts when run on a blank tape.

Therefore $M$ halts when run on a blank tape
iff there exists an $n \geq 1$ such that $a^n
\Longrightarrow^* b^n$, completing the reduction from the halting
problem.  Since the halting problem is undecidable,
\texttt{REWRITE-POWER} is also undecidable.  \qed
\end{proof}

\section{An undecidable problem on finite automata}
\label{three}

Our model of finite automaton is the usual one (e.g., 
\cite{Hopcroft&Ullman:1979}).
We now consider a decision problem on finite automata.
To state it, we need the notion of the product of two words
of the same length.  Let $\Sigma, \Delta$ be alphabets, and let
$w \in \Sigma^*$, $x \in \Delta^*$, with $|w| = |x|$.  Then
by $w \times x$ we mean the word 
over the alphabet $\Sigma \times \Delta$
whose projection $\pi_1$
over the first coordinate is $w$ and whose projection $\pi_2$
over the second coordinate is $x$.  More precisely, if 
$w = a_1 a_2 \cdots a_n$ and
$x = b_1 b_2 \cdots b_n$, then
$$w \times x = [a_1, b_1] [a_2, b_2] \cdots [a_n, b_n]. $$  
In this case $\pi_1 (w\times x) = w$ and
$\pi_2(w\times x) = x$.  For example,
if 
$$y = [{\tt t,h}] \, [{\tt e,o}] \, [{\tt r,e}] \, [{\tt m,s}],$$
then $\pi_1(y) = {\tt term}$ and
$\pi_1 (y) = {\tt hoes}$.
To simplify notation, we often write
$\atopp{w}{x}$
in place of $w \times x$.  For example, 
$$ \left[ \begin{array}{cc}
{\tt cat} \\
{\tt dog} \end{array} \right] \quad \text{ means the same thing as }
	\quad \atopp{\tt c}{\tt d} \atopp{\tt a}{\tt o} \atopp{\tt t}{\tt g}
	\quad \text{ and } \quad [{\tt c},{\tt d}] [{\tt a},{\tt o}] 
	[{\tt t}, {\tt g}] .$$

Consider the following decision problem:

\bigskip

\noindent {\tt ACCEPTS-SHIFT}

\medskip

\noindent {\it Instance:}
an alphabet $\Gamma$, a letter $c \not\in \Gamma$, and a finite
automaton $M$ with input alphabet $(\Gamma \union \{c\})^2$.

\medskip

\noindent 
{\it Question:}
does $M$ accept a word of the form $x c^n \times c^n x$
for some $x \in \Gamma^*$ and $n \geq 0$?

\begin{theorem}
The decision problem {\tt ACCEPTS-SHIFT} is undecidable.
\label{asu}
\end{theorem}

\begin{proof}
We reduce from the problem {\tt REWRITE-POWER}.  
An instance of this decision problem is a set $S$ of
length-preserving rewriting rules, an alphabet $\Sigma$, and letters
$a, b \in \Sigma$.
Define $\Gamma = \Sigma \ \cup \ \{ d \}$, where $d \not\in \Sigma$
is a new symbol.
Now we define the following regular languages:
\begin{align*}
E & = \left\{ \atopp{e}{e} \ : \ e \in \Sigma \right\} \\[1em]
R & = \left\{ \atopp{r}{\ell} \ : \ 
	\ell \rightarrow r \in S \right\} \\[1em]
L &= 
\atopp{d}{c}
{\atopp{a}{c}}^+
\atopp{d}{d}
\left( E^* R E^* 
\atopp{d}{d}
\right)^*
{\atopp{c}{b}}^+
\atopp{c}{d} .
\end{align*}
Let $M = (Q, \Delta, \delta, q_0, F)$ be a deterministic finite
automaton accepting $L$, with
$\Delta = (\Gamma \ \cup \ \{c\})^2$.
Clearly $M$ can be constructed effectively
from the definitions.   

We claim that, for all $n \geq 2$, we have
$a^{n-1} \derivestar b^{n-1}$ 
iff 
the language $L = L(M)$ contains a word of the form $x c^n \times c^n x$.
The crucial observation is that
\begin{equation}
u \derives v  \quad \text{iff} \quad
\atopp{v}{u}
\in
E^* R E^*.
\label{cruc}
\end{equation}
This follows immediately from the definitions of $E$ and $R$.

\bigskip

\noindent $\Longrightarrow$: 
Suppose
$$ u_0 := a^{n-1}  \derives u_1 \derives \cdots \derives
u_m = b^{n-1}$$
with $m \geq 1$ and $n \geq 2$.
Then 
$$ \atopp{u_{i+1}}{u_i} \in E^* R E^* $$
for $0 \leq i < m$.
Then
$$ 
\atopp{d}{d}
\atopp{u_1}{u_0} \atopp{d}{d} \atopp{u_2}{u_1} \cdots
\atopp{d}{d} \atopp{u_m}{u_{m-1}} \atopp{d}{d} \in 
\atopp{d}{d}
\left( E^* R E^* 
\atopp{d}{d}
\right)^* . \\[1em]
$$
Hence
$$ \atopp{d}{c} \atopp{u_0}{c^{n-1}} \atopp{d}{d}
\atopp{u_1}{u_0} \atopp{d}{d} \atopp{u_2}{u_1} \cdots
\atopp{d}{d} \atopp{u_m}{u_{m-1}} \atopp{d}{d}
\atopp{c^{n-1}}{u_m} \atopp{c}{d} \in L, \\[1em]$$
as desired.
The first component is
$d u_0 du_1 d \cdots d u_m d c^n$, while the
second component is
$c^n d u_0 d u_1 \cdots d u_{m-1} d u_m d$.
Taking 
$x = d u_0 d u_1 d \cdots d u_m d,$
we see that $x c^n \times c^n x \in L$.

\bigskip

\noindent $\Longleftarrow$:  Assume that $x c^n \times c^n x \in L$ for
some word $x$ with $n \geq 2$.  Now $L$ consists only of words of the form
$$
w = \atopp{d}{c}
{\atopp{a}{c}}^i
\atopp{d}{d}
\atopp{v_0}{u_0} 
\atopp{d}{d}
\atopp{v_1}{u_1}
\atopp{d}{d}
\cdots
\atopp{v_m}{u_m}
\atopp{d}{d}
{\atopp{c}{b}}^j
\atopp{c}{d} \\[1em]
$$
where $u_t \derives v_t$ for $1\leq t\leq m$ and $i, j \geq 1$.
Observe that
\begin{align*}
\pi_1 (w) &= d a^i d v_0 d v_1 \cdots d v_m d c^{j+1} \\
\pi_2 (w) &= c^{i+1} d u_0 d u_1 \cdots d u_m d b^j d,
\end{align*}
so if $\pi_1 (w) = xc^n$ and $\pi_2 (w) = c^n x$ we must have
$i = j = n-1$  and $x = da^i d v_0 d v_1 \cdots d v_m d 
= d u_0 d u_1 \cdots d u_m d b^j d$.
Since $d$ is a new symbol, not in the alphabet of $\Sigma$,
it follows that 
$u_0 = a^i$, $u_1 = v_0$, $u_2 = v_1, \ldots$, $u_m = v_{m-1}$, and
$ b^j = v_m $.

But then $u_0 \derives v_0 = u_1$,
$u_1 \derives v_1 = u_2$, and so forth, up to 
$u_{m-1} \derives v_{m-1} = u_m$, and finally
$u_m \derives v_m = b^j$.  So $u_0 \derivestar v_m$, and therefore
$a^{n-1} \derivestar b^{n-1}$.  This completes the proof. \qed
\end{proof}

\begin{remark}
In the decision problem {\tt ACCEPTS-SHIFT},
the undecidability of the problem arises, in an essential way,
from words of the form
$x c^n \times c^n x$ where $n < |x|$, and not from those words
with $n \geq |x|$ as one might first suspect.  More formally,
the related decision problem defined below is
actually solvable in cubic time.

\bigskip

\noindent {\tt ACCEPTS-LONG-SHIFT}

\medskip

\noindent {\it Instance:}
an alphabet $\Sigma$, a letter $c \not\in \Sigma$, and a finite
automaton $M$ with input alphabet $(\Sigma \union \{c\})^2$.

\medskip

\noindent 
{\it Question:}
does $M$ accept a word of the form $x c^n \times c^n x$
for some $n \geq |x|$?

\begin{theorem}
The decision problem {\tt ACCEPTS-LONG-SHIFT} is solvable
in cubic time.
\end{theorem}

\begin{proof}
Suppose $x = a_1 a_2 \cdots a_m$.
If $y = x c^n \times c^n x$ and $n \geq |x|$ then
$$y = [a_1, c] \cdots [a_m, c] [c,c]^{n-m} [c, a_1] [c, a_2] 
\cdots [c, a_m] .$$
Given a DFA $M = (Q, \Sigma, \delta, q_0, F)$,
we can create a nondeterministic finite automaton $M'$ that accepts all $x$
for which the corresponding $y$ is accepted by $M$.
The idea is that $M'$ has state set $Q' = Q\times Q \times Q$; on input
$x$ the machine $M'$ ``guesses'' a state $q \in Q$, and stores
it in the second component, and then simulates $M$ on input
$x \times c^m$ in the first component, starting from $q_0$ and
reaching some state $p$,
and simulates $M$ on input
$c^m \times x$ in the third component, starting from $q$.
Finally, $M'$ accepts if the third component is an element of
$F$ and if there exists a path from $p$ to $q$ labeled
$[c,c]^i$ for some $i \geq 0$.  Now we can test whether
$M'$ accepts a word by using depth-first or breadth-first
search on the transition diagram of $M'$, whose size is at most
cubic in terms of the size of $M$. \qed
\end{proof}
\end{remark}

\section{Application to $k$-automatic sets of rational numbers}
\label{four}

Recently the second author and co-authors defined a notation
of $k$-automaticity for sets of non-negative rational numbers 
\cite{Schaeffer&Shallit:2012,Rowland&Shallit:2015},
in analogy with the more well-known concept for sets of
non-negative integers \cite{Cobham:1972}.

For an integer $k \geq 2$ define $\Sigma_k = \lbrace
0,1, \ldots, k-1 \rbrace$.    If $w \in \Sigma_k^*$, define
$[w]_k$ to be the integer represented by the word $w$ in
base $k$ (assuming the most significant digit is at the left).
Let $M$ be a finite automaton with input alphabet
$\Sigma_k \times \Sigma_k$.
We define $\quo_k (M) \subseteq \Que^{\geq 0}$ to be the set 
$$\left\lbrace {{[\pi_1 (x)]_k} \over {[\pi_2 (x)]_k}} \ : \
	x \in L(M) \right\rbrace.$$
Furthermore, we call a set $T \subseteq \Que^{\geq 0}$ $k$-automatic
if there exists a finite automaton $M$ 
such that $T = \quo_k (M)$.

We first consider the following decision problem:

\bigskip

\noindent {\tt ACCEPTS-POWER}

\medskip

\noindent{\it Instance:} an integer $k \geq 2$, and 
a finite automaton $M$ with input
alphabet $(\Sigma_k)^2$.

\medskip

\noindent{\it Question:}  Is $\quo_k (L(M)) \intersect \{ k^i \ : \ i \geq 0 \}$
nonempty?

\begin{theorem}
The problem {\tt ACCEPTS-POWER} is undecidable.
\end{theorem}

\begin{proof}
The basic idea is to reduce once more from {\tt REWRITE-POWER},
using the same construction as in the proof of Theorem~\ref{asu}.
Our reduction
produces an instance of {\tt ACCEPTS-SHIFT} consisting
of an alphabet $\Gamma$ of cardinality $\ell$, a letter $c \not\in \Gamma$,
and a finite automaton $M$.
By renaming symbols, if necessary,
we can assume the symbols of $\Gamma$ are the digits $1,2,\ldots, \ell$
and $c$ is the digit $0$.  It then suffices to take $k = \ell + 1$.
Then $y \in L(M)$ with $\quo_k (y)$ a power of $k$ if and only
if $y = x 0^n \times 0^n x$ for some $x$ and some $n \geq 0$.
Note that, by our construction in the proof of Theorem~\ref{asu},
if $M$ accepts $x0^n \times 0^n x$, then $x$ contains no $0$'s. \qed
\end{proof}

Now consider a family of analogous decision problems
{\tt ACCEPTS-POWER}($k$), where in each problem $k$ is fixed.

\begin{theorem}
For each integer $k \geq 2$, the decision problem
{\tt ACCEPTS-POWER}$(k)$ is undecidable.
\end{theorem}

\begin{proof}
We have to overcome the problem that $k$ can depend on 
the size of $\Gamma$.  To do so, 
we recode all
words over the alphabet $\{0, 1\}$.  
It suffices to use the morphism
$\varphi$ defined by
\begin{align*}
\varphi(c) &    = 0^{m+1} \quad & \varphi(a_i) &= 1^i 0^{m-i} 1
\end{align*}
where $\Sigma = \{ a_1, a_2, \ldots, a_m \}$.  
In the proof of Theorem~\ref{asu}, we replace
$E, R, L$ by $E', R', L'$, as follows:
\begin{align*}
E' & = \left\{ \atopp{\varphi(e)}{\varphi(e)} \ : \ e \in \Sigma \right\} \\[1em]
R' & = \left\{ \atopp{\varphi(r)}{\varphi(\ell)} \ : \ 
	\ell \rightarrow r \in S \right\} \\[1em]
L' &= 
\atopp{\varphi(d)}{\varphi(c)}
{\atopp{\varphi(a)}{\varphi(c)}}^+
\atopp{\varphi(d)}{\varphi(d)}
\left( E^* R E^* 
\atopp{\varphi(d)}{\varphi(d)}
\right)^*
{\atopp{\varphi(c)}{\varphi(b)}}^+
\atopp{\varphi(c)}{\varphi(d)} .
\end{align*}
The construction works because the blocks for symbols of $\Sigma$
begin and end with at least one $1$, while the block for $c$ consists
of all $0$'s. Therefore, if the first coordinate of an element of
$L'$ has a suffix in $0^+$, this can only arise from $\varphi(c)$,
and the same for prefixes of the second coordinate. \qed
\end{proof}

\section{Problems about conjugates}
\label{five}

Recall that we say two words $x$ and $y$ are {\it conjugates} if
one is a cyclic shift of the other; that is, if there exist $u, v$
such that $x = uv$ and $y = vu$.

The undecidability result of the previous section suggests studying the
following related natural decision problem.

\medskip

\noindent{\tt ACCEPTS-GENERAL-SHIFT}

\medskip

\noindent{\it Instance:} a finite automaton $M$ with input alphabet $\Sigma^2$.

\medskip

\noindent{\it Question:} does $M$ accept a word of the form
$x \times y$ for some conjugates $x,y \in \Sigma^*$ ?

\begin{theorem}
The decision problem {\tt ACCEPTS-GENERAL-SHIFT} is undecidable.
\end{theorem}

\begin{proof}
We reduce from the problem {\tt ACCEPTS-SHIFT}.  An instance of this
problem is an alphabet $\Gamma$, a letter $c \notin \Gamma$, and
a finite automaton $M$ with input alphabet $(\Gamma \cup \{c\})^2$.

First check whether $M$ accepts a word of the form $x \times x$ for some $x \in \Gamma^*$.  (This is decidable because the language 
$\{x \times x \ : \ x \in \Gamma^*\}$ is regular.)
If so, {\tt ACCEPTS-SHIFT}$(\Gamma,c,M)$ = ``yes''.
Otherwise, construct a finite automaton $M'$ whose language is
\begin{align*}
L(M) \cap \{sc^+ \times c^+t \mid s,t \in \Gamma^*\}.
\end{align*}
Notice that
{\tt ACCEPTS-SHIFT}$(\Gamma,c,M')$ = {\tt ACCEPTS-SHIFT}$(\Gamma,c,M)$.
Clearly we have that if 
$$ \mbox{\tt ACCEPTS-GENERAL-SHIFT}(M') = \text{``no''}, $$
then {\tt ACCEPTS-SHIFT}$(\Gamma,c,M')$ = ``no''.  

So suppose that {\tt
ACCEPTS-GENERAL-SHIFT}$(M')$ = ``yes''.  Then $M'$ accepts a word $w = x
\times y$ for words $x = uv$, $y = vu$ where
$u,v \in (\Gamma \cup \{c\})^*$.  We now show
that $w = zc^n \times c^nz$ for some $z \in \Gamma^*$ and $n \geq 1$.

By the construction of $M'$, $uv$ ends with $c$, $vu$ begins with $c$,
and any two occurrences of $c$ in $uv$ or $vu$ have only $c$'s between
them.  Hence if $u$ or $v$ is empty, then $w = c^n \times c^n$ for some
$n \geq 1$, and we can take $z = \epsilon$, the empty word.  So say
neither $u$ nor $v$ is empty.  Then $v$ begins and ends with $c$, and hence
$v$ is in $c^+$.  It follows that if $u$ contains $c$, then $u$ begins
and ends with $c$, so again $w = c^n \times c^n$ for some $n \geq 1$,
and we can take $z = \epsilon$.  So say $u$ does not contain $c$.  Then
$w = uc^n \times c^n u$ with $u \in \Gamma^+$ and $n = |v|$, and we can
take $z = u$.

So $w = zc^n \times c^nz$ for some word $z \in \Gamma^*$ and $n \geq 1$. 
Therefore we have {\tt ACCEPTS-SHIFT}$(\Gamma,c,M')$ = ``yes''.  This completes
the reduction.  Then since {\tt ACCEPTS-SHIFT} is undecidable by
Theorem~\ref{asu}, {\tt ACCEPTS-GENERAL-SHIFT} is also undecidable.  \qed
\end{proof}

Now we turn to two other decision problems, both inspired by
the problem {\tt ACCEPTS-GENERAL-SHIFT}.  The first is

\newpage

\noindent{\tt ACCEPTS-DISTINCT-CONJUGATES}

\medskip

\noindent{\it Instance:}  A DFA $M = (Q, \Sigma, \delta, q_0, F)$.

\medskip

\noindent{\it Question:}  Does $M$ accept two distinct conjugates
$uv$ and $vu$?

\medskip

We will prove

\begin{theorem}
{\tt ACCEPTS-DISTINCT-CONJUGATES} is decidable.
\label{adc-thm}
\end{theorem}

To prove this theorem, we need the concept of primitive word and primitive
root.  A nonempty word $x$ is said to be {\it primitive}
if it cannot be written in the form $x = y^i$ for a word $y$ 
and an integer $i \geq 2$.  The {\it primitive root} of a word $x$
is the unique primitive word $t$ such that $x = t^j$ for some $j \geq 1$.

\begin{lemma} 
If a DFA $M$ of $n$ states accepts two distinct conjugates, then 
it accepts two distinct conjugates $uv$ and $vu$,
with at least one of $u$ and $v$ of length $\leq n^2$.
\label{conjlem}
\end{lemma}

\begin{proof}
Let $L = L(M)$, the language accepted by $M = (Q, \Sigma, \delta, q_0, F)$,
where $|Q| = n$.
Suppose that there exist $uv \in L$, $vu \in L$, but $uv \not= uv$.
Without loss of generality, assume $|uv|$ is as small as possible.
Assume, contrary to what we want to prove, that both $|u|$ and $|v|$
are $> n^2$.    

Consider the acceptance path of $uv$ through $M$:  it looks like
$\delta(q_0, u) = q_1$ and $\delta(q_1, v) = p_1$ for some $q_1 \in Q$ and
$p_1 \in F$.   Similarly, consider the acceptance path of $vu$ through
$M$:  it looks like $\delta(q_0, v) = q_2$ and $\delta(q_2, u) = p_2$ for
some $q_2 \in Q$ and $p_2 \in F$.

Now create a new DFA 
$M' = (Q\times Q, \Sigma, \delta', q'_0, F')$
by the usual product construction,
where $\delta'([r,s],a) := [\delta(r,a), \delta(s,a)]$
and $q'_0 = [q_0, q_1]$ and $F = \{ [q_2, p_1] \}$.
Then $M'$ has $n^2$ states and accepts $v$. 

Since $|v| > n^2$,
the acceptance path for $v$ in $M'$ visits $\geq n^2 + 2$ states and
hence some state is repeated, giving us a loop of at most $n^2$ states
that can be cut out.
Hence we can write $v = v_1 v_2 v_3$, where $v_2 \not= \epsilon$ and
$v_1 v_3 \not= \epsilon$,
and $M'$ accepts $v_1 v_3$.  In $M$, then, it follows that
$\delta(q_1, v_1 v_3) = p_1$ and
$\delta(q_0, v_1 v_3) = q_2$, and hence $M$
accepts the conjugates $u v_1 v_3$ and $v_1 v_3 u$.  Since $|uv_1v_3|< |uv|$,
the minimality of $|uv|$ implies that these conjugates cannot be distinct,
and so we must have
\begin{equation}
uv_1 v_3 = v_1 v_3 u .
\label{u13}
\end{equation}

We can now repeat the argument of the previous paragraph for the 
word $u$.  We get a decomposition $u = u_1 u_2 u_3$
where $u_2 \not= \epsilon$ and $u_1 u_3 \not= \epsilon$, and we get 
\begin{equation}
v u_1 u_3 = u_1 u_3 v.
\label{v13}
\end{equation}

Finally, the acceptance paths in $M$ we have created imply that we
can cut out both $u_2$ and $v_2$ simultaneously from $uv$ and $vu$,
and still get words accepted by $M$.  So $u_1 u_3 v_1 v_3$ and
$v_1 v_3 u_1 u_3$ are both accepted.   Again, by minimality, we get
that
\begin{equation}
u_1 u_3 v_1 v_3 = v_1 v_3 u_1 u_3 .
\label{uv13}
\end{equation}

Now, by the Lyndon-Sch\"utzenberger
theorem (see, e.g., \cite{Lyndon&Schutzenberger:1962,Shallit:2009}),
Eq.~\eqref{uv13} implies the existence of
a nonempty word $t$ and integers
$i, j$ such that $u_1 u_3 = t^i$, $v_1 v_3 = t^j$.   Without loss of generality,
we can assume that $t$ is primitive.

Applying the same theorem to Eq.~\eqref{v13} tells us that there exists
$k$ such that $v = t^k$.  And applying the same theorem once more to
Eq.~\eqref{u13} tells us that there exists $\ell$ such that
$u = t^\ell$.  But then $uv = vu$, a contradiction.  \qed
\end{proof}

\begin{remark}
We observe that the bound of $n^2$ in the previous result is optimal,
up to a constant multiplicative factor.  Consider the languages
$$L_t = (a^t)^+ b (a^{t+1})^+ bb \ \cup \ 
(a^t)^+ bb (a^{t+1})^+ bb .$$
Then it is easy to see that $L_t$ can be accepted
by a (complete) DFA of $n = 3t+8$ states.  The shortest pair of distinct
conjugates in $L_n$, however, are
$a^{t(t+1)} b a^{t(t+1)} bb$ and
$a^{t(t+1)} bb a^{t(t+1)} b$, corresponding to
$u = a^{t(t+1)} b$ of length $t^2 + t + 1$ and
$v = a^{t(t+1)} bb$ of length $t^2 + t + 2$.   Thus both $u$ and $v$
are of length $n^2/9 + O(n)$.  
\end{remark}

We can now prove Theorem~\ref{adc-thm}.

\begin{proof}
Given $L = L(M)$,
for each nonempty word $x$ define the language
$$ L_x = \{ y \in\Sigma^* \ : \ xy \in L,\ yx \in L,\ xy \not= yx \}.$$

We observe that each $L_x$ is a regular language.  To see this,
note that we can write
$L_x = L_1 \, \cap \, L_2 \, \cap \, L_3$, where
\begin{align*}
L_1 &= \{ y \in \Sigma^* \ : \ xy \in L \} \\
L_2 &= \{ y  \in \Sigma^* \ : \ yx \in L \} \\
L_3 &= \{ y \in \Sigma^* \ : \ xy \not= yx \}  .
\end{align*}

Both $L_1$ and $L_2$ are easily seen to be regular, and finite
automata accepting
them are easily constructed from $M$.
To see that the same holds for $L_3$, note that if $xy = yx$ with
$x$ nonempty, then by the Lyndon-Sch\"utzenberger theorem it follows
that $y \in t^*$, where $t$ is the primitive root of $x$.
Hence $L_3 = \overline{t^*}$.  Therefore we can construct a finite automaton
$M_x$ accepting $L_x$.

Finally, here is the decision procedure.  
By Lemma~\ref{conjlem} we know that if an $n$-state DFA
$M$ accepts a pair of words $uv$ and $vu$ with $uv \not= vu$,
then it must accept a pair with either $|u| \leq n^2$ or
$|v| \leq n^2$.   Thus, it suffices to enumerate
all $u \in \Sigma^*$ of lengths $1, 2, \ldots, n^2$, and
compute $M_u$ for each $u$.  If at least one $M_u$ has $L(M_u)$ nonempty,
then answer ``yes''; otherwise answer ``no''.  \qed
\end{proof}

Our second decision problem is

\bigskip

\noindent{\tt ACCEPTS-NON-CONJUGATES}

\medskip

\noindent{\it Instance:}  A DFA $M = (Q, \Sigma, \delta, q_0, F)$.

\medskip

\noindent{\it Question:}  Does $M$ accept two words of the same
length that are not conjugates?

\bigskip

We prove

\begin{theorem}
{\tt ACCEPTS-NON-CONJUGATES} is decidable.
\end{theorem}

\begin{proof}
Given a formal language $L$ over an ordered alphabet $\Sigma$,
we define $\lexleast(L)$ to be the union, over all
$n \geq 0$, of the lexicographically least word of length $n$ in $L$,
if it exists.
As is well-known (see, e.g., \cite[Lemma 1]{Shallit:1994}), if $L$ is
regular, then so is $\lexleast(L)$.  Furthermore, given a DFA
for $L$, we can algorithmically construct a DFA for $\lexleast(L)$.

We also define $\cyc(L)$ to be the union, over all words $w \in L$,
of the conjugates of $w$.  Again, as is well-known (see, e.g.,
\cite[Thm.~3.4.3]{Shallit:2009}), if $L$ is regular, then so is
$\cyc(L)$.  Furthermore, given a DFA for $L$, we can algorithmically
construct a DFA for $\cyc(L)$.

We claim that $L$ contains two words $x$ and $y$ of the same length
that are non-conjugates if and only if $L$ is not a subset of
$\cyc(\lexleast(L))$.  

Suppose such $x, y$ exist.  Let $t$ be the lexicographically least
word in $L$ of length $|x|$.  If $t$ is a conjugate of 
$x$, then $y$ is not a conjugate of $t$,
so $y \not\in \cyc(\lexleast(L))$.  On the other hand, if $t$
is not a conjugate of $x$, then $x \not\in \cyc(\lexleast(L))$.  In both
cases $L$ is not a subset of $\cyc(\lexleast(L))$.

Suppose $L$ is not a subset of $\cyc(\lexleast(L))$.  Then there is
some word of some length $n$ in $L$, say $x$, that is not a 
conjugate of the lexicographically least word of length $n$, say $y$.
Then $x$ and $y$ are the desired two words.

Putting this all together, we get our decision procedure for the
decision problem
{\tt ACCEPTS-NON-CONJUGATES}: given the DFA $M$ for $L$,
construct the DFA $M'$ for
$L - \cyc(\lexleast(L)) $ using the techniques mentioned above.
If $M'$ accepts at least one word, then the answer for
{\tt ACCEPTS-NON-CONJUGATES} is ``yes''; otherwise it is ``no''.  \qed
\end{proof}

\section{Final remarks}
\label{seven}

We still do not know whether the following problem
from \cite[p.~363]{Rowland&Shallit:2015}
is decidable:

\bigskip

\noindent{\tt ACCEPTS-INTEGER}

\medskip

\noindent{\it Instance:}  a finite automaton $M$ with input 
alphabet $(\Sigma_k)^2$.

\medskip

\noindent{\it Question:}  Is
$\quo_k (L(M)) \ \cap \ \Enn$ nonempty?

\bigskip

\noindent Unfortunately our techniques do not seem immediately applicable to this
problem.

\medskip

We mention two other problems about finite automata whose decidability
is still open:

\medskip

1.  Given a DFA $M$ with input alphabet $\{ 0, 1\}$, decide if there exists
at least one prime number $p$ such that $M$ accepts the base-$2$
representation of $p$.

\begin{remark}
An algorithm for this problem would allow
resolution of the existence of a Fermat prime $2^{2^k} + 1$ for
$k > 4$.
\end{remark}

\medskip

2. Given a DFA $M$ with input alphabet $\{0, 1\}$, decide if there
exists at least one integer $n\geq 0$ such that $M$ accepts the
base-$2$ representation of $n^2$.

\section{Acknowledgments}

We thank Hendrik Jan Hoogeboom for his helpful comments.

\end{document}